\documentclass[letterpaper, 10 pt, conference]{ieeeconf}  

\pdfoutput=1

\IEEEoverridecommandlockouts                              
\usepackage{amsmath} 
\usepackage{amssymb}  
\usepackage{mathtools} 		
\usepackage[english, status=final, layout=inline]{fixme}
\fxusetheme{color}
\usepackage{overpic}
\usepackage{tikz}
\usetikzlibrary{arrows.meta,calc}
\usepackage{graphicx}
\usepackage{colortbl}  
\usepackage{siunitx}  
\usepackage{overpic}
\usepackage{subcaption} 
\usepackage{cite}  
\usepackage[hyphens]{url}

\newtheorem{theorem}{Theorem}
\newtheorem{definition}{Definition}
\newtheorem{lemma}{Lemma}
\newtheorem{remark}{Remark}

\definecolor{matplotlibBlue}{rgb}{0.12156862745098039, 0.4666666666666667, 0.7058823529411765}%
\definecolor{matplotlibGray}{rgb}{0.4980392156862745, 0.4980392156862745, 0.4980392156862745}%
\definecolor{matplotlibGreen}{HTML}{2ca02c}%
\definecolor{matplotlibRed}{HTML}{d62728}%

\newcolumntype{a}{>{\columncolor{matplotlibGreen!60}}c}

\title{\LARGE \bf
	A Hierarchical Attack Identification Method for Nonlinear Systems
}

\author{Sarah Braun$^{1, 2}$ and Sebastian Albrecht$^{1}$ and Sergio Lucia$^{2, 3}$
	\thanks{$^{1}$Sarah Braun and Sebastian Albrecht are with Siemens Corporate Technology, Otto-Hahn-Ring 6, 81739 Munich, Germany. {Email: \tt\small sarah.braun@siemens.com}}%
	\thanks{$^{2}$Sarah Braun and Sergio Lucia are with the Laboratory of Internet of Things for Smart Buildings, Technische Universit\"at Berlin, Ernst-Reuter-Platz 7, 10587 Berlin, Germany.}%
	\thanks{$^{3}$Sergio Lucia is with Einstein Center Digital Future, Wilhelmstra{\ss}e 67, 10117 Berlin, Germany.}%
	\thanks{This work was supported by the German Federal Ministry of Education and Research (BMBF) via the funded research project AlgoRes (01S18066B).}
}

\begin{document}
	\onecolumn
	\textcopyright 2020 IEEE. Personal use of this material is permitted. Permission from IEEE must be obtained for all other uses, in any current or future media, including reprinting/republishing this material for advertising or promotional purposes, creating new collective works, for resale or redistribution to servers or lists, or reuse of any copyrighted component of this work in other works.
	
	\twocolumn
	\maketitle
	\thispagestyle{empty}
	\pagestyle{empty}

\begin{abstract}
Many autonomous control systems are frequently exposed to attacks, so methods for attack identification are crucial for a safe operation.
To preserve the privacy of the subsystems and achieve scalability in large-scale systems, identification algorithms should not require global model knowledge. 
We analyze a previously presented method for hierarchical attack identification, that is embedded in a distributed control setup for systems of systems with coupled nonlinear dynamics. It is based on the exchange of local sensitivity information and ideas from sparse signal recovery. In this paper, we prove sufficient conditions under which the method is guaranteed to identify all components affected by some unknown attack. 
Even though a general class of nonlinear dynamic systems is considered, our rigorous theoretical guarantees are applicable to practically relevant examples, which is underlined by numerical experiments with the IEEE~30 bus power system.
\end{abstract}

\section{INTRODUCTION}
The control of dynamic systems in safety-critical infrastructures such as power systems, factory automation or traffic networks has been automated more and more over the last decades. While the increasing degree of automation involves opportunities to improve the system's efficiency and integrity, it further increases the threat of malicious attacks
on physical or cyber components of the system. 
It is therefore crucial to develop methods for preventing, identifying, and handling attacks. The communication layers of cyber-physical systems are protected by means of IT 
security, and also the system's resilience on the control layer can be increased, e.g., by robust control. Nevertheless, absolute safety cannot be guaranteed. Therefore, each autonomous system should be equipped with methods for attack detection and identification to reveal the existence and location of an attack.

We consider a networked control system with states $x \in \mathbb{X} \subseteq \mathbb{R}^{d_x}$, initial state $x^0 \in \mathbb{X}$ and control $u\in \mathbb{U} \subseteq \mathbb{R}^{d_u}$, that consists of a set $\mathcal{P}$ of physically coupled subsystems with nonlinear dynamics. The dynamics of the system are exposed to possible attacks, where an \emph{attack} is modeled as a modification $a(u) \neq u$ of the input $u \in \mathbb{U}$ through an \emph{attack function} $a: \mathbb{U} \rightarrow \mathbb{U}$. Modeling an attack as a disturbance in the input is a frequently used attack model, \mbox{see~\cite{Pasqualetti2013Attack,Giraldo2018Survey,Ananduta2020Resilient}}, and implies that the intended controller action does not match the actual actuation of the system~\cite{Giraldo2018Survey}. While controller or actuator attacks are thus clearly covered by the attack model, also sensor attacks can be expressed by suitable attack functions since a sensor can be modeled as a simple input-output device. An attack can alter the local inputs $u_I \in \mathbb{U}_I$ in one or several subsystems $I\in \mathcal{P}$, and modify one or multiple input components $(u_I)_i$. It may or may not depend on the undisturbed control $u$ and we do \emph{not} assume the set of possibly \mbox{occurring attacks nor any attack patterns to be known. }

Denoting the local states of subsystem $I$ by $x_I \in \mathbb{X}_I$, the nonlinear discrete-time dynamics of subsystem $I$ including possible, unknown attacks $a_I$ are given as
\begin{equation}
\label{eq: local dynamics}
\begin{aligned}
x_I^+ &= f_I(x_I, a_I(u_I), z_{\mathcal{N}_I}), \\ 
z_I &= h_I(x_I).
\end{aligned}
\end{equation}
The function $h_I$ relates the local states $x_I$ to the local coupling variables $z_I \in \mathbb{R}^{d_{z_I}}$ through which subsystem $I$ influences other subsystems. By $\mathcal{N}_I$ we denote the neighborhood of subsystem $I$, that is defined as the set of all subsystems~$J$ influencing the dynamics of $x_I$ through couplings $z_J$. 
\subsection{Related Work}
A series of recently published surveys shows comprehensive research on control and model-based approaches towards attack detection and identification in cyber-physical systems~\cite{Giraldo2018Survey,Ding2018survey,Dibaji2019Systems}. 
Many proposed methods involve observer-based filters that are tailored for linear dynamics, e.g.,~\cite{Pasqualetti2013Attack,Gallo2018Distributed,Shames2011Distributed,Ding2008Model}. Both centralized~\cite{Ding2018survey} and distributed~\cite{Gallo2018Distributed,Pasqualetti2012Consensus} filters requiring only local model knowledge exist. 
Similar to our approach, some methods involve optimization problems to compute plausible sparse attack signals~\cite{Liu2014Detecting} or update the probability of hypotheses on the attack constellation~\cite{Ananduta2020Resilient}.
Some papers deal with networked systems with special properties such as consensus networks or weakly coupled subsystems~\cite{Pasqualetti2012Consensus}, other frameworks depend on the attackers' resources~\cite{Teixeira2015secure}. While some of these methods for linear systems have been applied to attack identification in power systems~\cite{Pasqualetti2013Attack,Shames2011Distributed}, using linearized swing equations to model the dynamics in power systems is only valid as long as the phase angles are close to each other~\cite{Pan2015Online}. Since this cannot be guaranteed in case of attack, identification methods designed for systems with nonlinear dynamics should be considered. To this end, de Persis and Isidori propose a differential-geometric characterization of attack identification in nonlinear systems~\cite{DePersis2001geometric}. They present solvability conditions in terms of an unobservability distribution and derive a detection filter. 
However, the proposed conditions result in a centralized approach that is unsuitable for large-scale systems.
In contrast, Esfahani et al.\ propose a scalable residual generator for nonlinear systems with additive attacks, which is based on solving a sequence of quadratic programs~\cite{Esfahani2012tractable}. The nonlinearities in the dynamics are not taken as part of the model but as disturbances following some known patterns, and a linear filter which is robust towards these disturbances is applied. An approach to attack identification in power systems with modeled nonlinearities is presented in~\cite{Pan2015Online}. Similar to our method, a sparse signal recovery problem is solved to find an attack signal explaining the observed behavior. While the authors consider several subsequent time steps under constant attack and apply linear regression requiring measurements of all phase angles, our approach uses measurements at some coupling nodes and one sampling time only. It can be classified as a hierarchical identification scheme since it requires aggregated sensitivity information but no global knowledge of the dynamics of each subsystem.

Further methods for linear and nonlinear systems can be found in the area of \emph{fault} detection and identification, which focuses on unintended system failures rather than malicious attacks~\cite{Boem2018Plug}. In this field, it is common to assume that the set of possible faults in nonlinear systems is known and finite, which is an invalid assumption for \emph{attack} identification~\cite{Ding2018survey}.

\subsection{Contribution}
We present a scalable attack identification method for distributed control systems in Section~\ref{sec: identification method}, which was introduced and successfully used to identify faulty buses in power systems in our preliminary work~\cite{Braun2020Hierarchical}. 
In contrast to, e.g., \cite{Pasqualetti2013Attack,Gallo2018Distributed,Shames2011Distributed,Esfahani2012tractable}, it is designed for explicitly modeled nonlinearities in the dynamics.
It involves the exchange of predicted nominal values for certain coupling states and local sensitivity information as in Fig.~\ref{fig: exchange nominal values}, based on which we approximate how an attack spreads through the network. Attack identification is then approached by solving a sparse signal recovery problem. While requiring the global knowledge of sensitivity information evaluated at the current iterate, the method does not involve the global dynamics nor cost functions nor measurements of all states, unlike~\cite{Pan2015Online,DePersis2001geometric,Teixeira2015secure}. It is designed for nonlinear dynamics but, in contrast to~\cite{Boem2018Plug}, does not assume all potential attacks to be known nor makes further restrictions like considering only additive attacks as in~\cite{Esfahani2012tractable}. 
\begin{figure}[b]
	\centering
	\begin{overpic}[tics=5,width=\columnwidth]{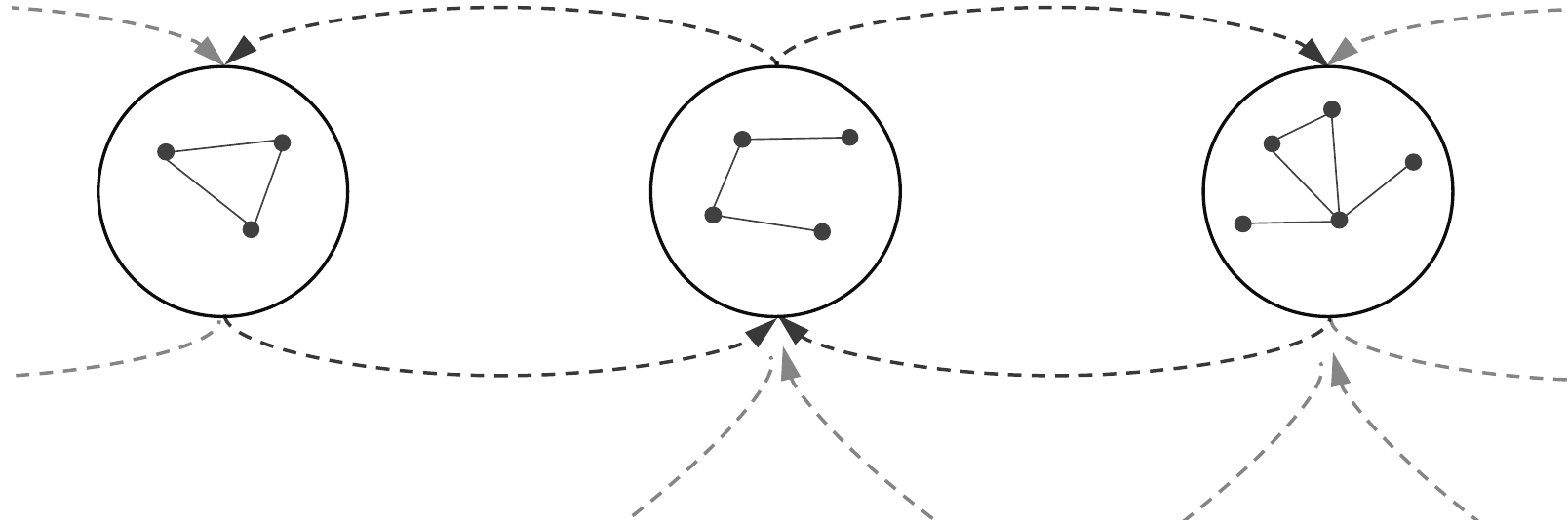}
		\put(12.7,14.3){II}
		\put(48.8,14.3){I}		
		\put(82.4,14.3){III}		
	\definecolor{matplotlibBlue}{rgb}{0.12156862745098039, 0.4666666666666667, 0.7058823529411765}%
	\definecolor{matplotlibGray}{rgb}{0.4980392156862745, 0.4980392156862745, 0.4980392156862745}%
	\def\DeltaT{0.45}%
	\def\AxOverlap{0.1}%
\newcommand{\contractOne}{%
\begin{tikzpicture}[scale=0.5]%
	\coordinate (O) at (0, -0.55);
	\coordinate (Z0) at (0, 0);
	\coordinate (Z1U) at (\DeltaT, 0.4);
	\coordinate (Z1N) at (\DeltaT, 0.13);
	\coordinate (Z1L) at (\DeltaT, -0.13);
	\coordinate (Z2U) at (2*\DeltaT, 0.45);
	\coordinate (Z2N) at (2*\DeltaT, 0.1);
	\coordinate (Z2L) at (2*\DeltaT, -0.3);
	\coordinate (Z3U) at (3*\DeltaT, 0.42);
	\coordinate (Z3N) at (3*\DeltaT, -0.05);
	\coordinate (Z3L) at (3*\DeltaT, -0.4);
	\coordinate (Z4U) at (4*\DeltaT, 0.3);
	\coordinate (Z4N) at (4*\DeltaT, -0.02);
	\coordinate (Z4L) at (4*\DeltaT, -0.4);
	\coordinate (Z5U) at (5*\DeltaT, 0.25);
	\coordinate (Z5N) at (5*\DeltaT, 0.02);
	\coordinate (Z5L) at (5*\DeltaT, -0.4);
	\coordinate (Z6U) at (6*\DeltaT, 0.3);
	\coordinate (Z6N) at (6*\DeltaT, 0.02);
	\coordinate (Z6L) at (6*\DeltaT, -0.35);
	
	\def\SignOne{0.13}
	\def\SignTwo{0.2}
	\def\SignThree{0.12}	
	\def\SignFour{0.15}	
	
	\draw [draw=matplotlibBlue!60, fill=matplotlibBlue!30, fill opacity=1.0] plot  coordinates {(Z0) (Z1L) (Z2L) (Z3L) (Z4L) (Z4U) (Z3U) (Z2U) (Z1U) (Z0)};
	
	\draw [draw=matplotlibBlue, thick] plot coordinates {(Z0) (Z1N) (Z2N) (Z3N) (Z4N)};
	
	\draw [{Bar[width=0.8mm]}-{Bar[width=0.8mm]}, draw=matplotlibBlue] ($(Z1N) - (0,\SignOne)$) -- ($(Z1N) + (0,\SignOne)$);
	\draw [{Bar[width=0.8mm]}-{Bar[width=0.8mm]}, draw=matplotlibBlue] ($(Z2N) - (0,\SignTwo)$) -- ($(Z2N) + (0,\SignTwo)$);
	\draw [{Bar[width=0.8mm]}-{Bar[width=0.8mm]}, draw=matplotlibBlue] ($(Z3N) - (0,\SignThree)$) -- ($(Z3N) + (0,\SignThree)$);
	\draw [{Bar[width=0.8mm]}-{Bar[width=0.8mm]}, draw=matplotlibBlue] ($(Z4N) - (0,\SignFour)$) -- ($(Z4N) + (0,\SignFour)$);		
		
	\draw [-{Latex[length=1.5mm]}, thick, black!60] ($(O)-(0,\AxOverlap)$) -- (0, 0.6);
	\draw [-{Latex[length=1.5mm]}, thick, black!60]  ($(O)-(\AxOverlap, 0)$) -- ($(O) + (5.3*\DeltaT, 0)$);
\end{tikzpicture}
}%
\newcommand{\contractTwo}{%
	\begin{tikzpicture}[scale=0.5]%
	\coordinate (O) at (0, -0.45);
	\coordinate (Z0) at (0, 0);
	\coordinate (Z1U) at (\DeltaT, 0.25);
	\coordinate (Z1N) at (\DeltaT, 0.05);
	\coordinate (Z1L) at (\DeltaT, -0.15);
	\coordinate (Z2U) at (2*\DeltaT, 0.2);
	\coordinate (Z2N) at (2*\DeltaT, -0.05);
	\coordinate (Z2L) at (2*\DeltaT, -0.3);
	\coordinate (Z3U) at (3*\DeltaT, 0.3);
	\coordinate (Z3N) at (3*\DeltaT, -0.02);
	\coordinate (Z3L) at (3*\DeltaT, -0.33);
	\coordinate (Z4U) at (4*\DeltaT, 0.4);
	\coordinate (Z4N) at (4*\DeltaT, 0.03);
	\coordinate (Z4L) at (4*\DeltaT, -0.3);
	\coordinate (Z5U) at (5*\DeltaT, 0.4);
	\coordinate (Z5N) at (5*\DeltaT, 0.06);
	\coordinate (Z5L) at (5*\DeltaT, -0.28);
	\coordinate (Z6U) at (6*\DeltaT, 0.38);
	\coordinate (Z6N) at (6*\DeltaT, 0.1);
	\coordinate (Z6L) at (6*\DeltaT, -0.31);
	
	\def\SignOne{0.12}
	\def\SignTwo{0.2}
	\def\SignThree{0.24}	
	\def\SignFour{0.15}	
	
	\draw [draw=matplotlibBlue!60, fill=matplotlibBlue!30, fill opacity=1.0] plot  coordinates {(Z0) (Z1L) (Z2L) (Z3L) (Z4L) (Z4U) (Z3U) (Z2U) (Z1U) (Z0)};
	
	\draw [draw=matplotlibBlue, thick] plot coordinates {(Z0) (Z1N) (Z2N) (Z3N) (Z4N)};
	
	\draw [{Bar[width=0.8mm]}-{Bar[width=0.8mm]}, draw=matplotlibBlue] ($(Z1N) - (0,\SignOne)$) -- ($(Z1N) + (0,\SignOne)$);
	\draw [{Bar[width=0.8mm]}-{Bar[width=0.8mm]}, draw=matplotlibBlue] ($(Z2N) - (0,\SignTwo)$) -- ($(Z2N) + (0,\SignTwo)$);
	\draw [{Bar[width=0.8mm]}-{Bar[width=0.8mm]}, draw=matplotlibBlue] ($(Z3N) - (0,\SignThree)$) -- ($(Z3N) + (0,\SignThree)$);
	\draw [{Bar[width=0.8mm]}-{Bar[width=0.8mm]}, draw=matplotlibBlue] ($(Z4N) - (0,\SignFour)$) -- ($(Z4N) + (0,\SignFour)$);
	
	\draw [-{Latex[length=1.5mm]}, thick, black!60] ($(O)-(0,\AxOverlap)$) -- (0, 0.6);
	\draw [-{Latex[length=1.5mm]}, thick, black!60] ($(O)-(\AxOverlap, 0)$) -- ($(O) + (5.3*\DeltaT, 0)$);	
	\end{tikzpicture}
}%
\newcommand{\contractThree}{%
\begin{tikzpicture}[scale=0.5]%
\coordinate (O) at (0, -0.65);
\coordinate (Z0) at (0, 0);
\coordinate (Z1U) at (\DeltaT, 0.25);
\coordinate (Z1N) at (\DeltaT, -0.02);
\coordinate (Z1L) at (\DeltaT, -0.27);
\coordinate (Z2U) at (2*\DeltaT, 0.35);
\coordinate (Z2N) at (2*\DeltaT, 0.05);
\coordinate (Z2L) at (2*\DeltaT, -0.25);
\coordinate (Z3U) at (3*\DeltaT, 0.40);
\coordinate (Z3N) at (3*\DeltaT, 0.05);
\coordinate (Z3L) at (3*\DeltaT, -0.3);
\coordinate (Z4U) at (4*\DeltaT, 0.43);
\coordinate (Z4N) at (4*\DeltaT, -0.06);
\coordinate (Z4L) at (4*\DeltaT, -0.4);
\coordinate (Z5U) at (5*\DeltaT, 0.45);
\coordinate (Z5N) at (5*\DeltaT, -0.08);
\coordinate (Z5L) at (5*\DeltaT, -0.45);
\coordinate (Z6U) at (6*\DeltaT, 0.41);
\coordinate (Z6N) at (6*\DeltaT, -0.06);
\coordinate (Z6L) at (6*\DeltaT, -0.5);

\def\SignOne{0.17}
\def\SignTwo{0.12}
\def\SignThree{0.24}	
\def\SignFour{0.27}	

\draw [draw=matplotlibBlue!60, fill=matplotlibBlue!30, fill opacity=1.0] plot  coordinates {(Z0) (Z1L) (Z2L) (Z3L) (Z4L)  (Z4U) (Z3U) (Z2U) (Z1U) (Z0)};

\draw [draw=matplotlibBlue, thick] plot coordinates {(Z0) (Z1N) (Z2N) (Z3N) (Z4N)};

\draw [{Bar[width=0.8mm]}-{Bar[width=0.8mm]}, draw=matplotlibBlue] ($(Z1N) - (0,\SignOne)$) -- ($(Z1N) + (0,\SignOne)$);
\draw [{Bar[width=0.8mm]}-{Bar[width=0.8mm]}, draw=matplotlibBlue] ($(Z2N) - (0,\SignTwo)$) -- ($(Z2N) + (0,\SignTwo)$);
\draw [{Bar[width=0.8mm]}-{Bar[width=0.8mm]}, draw=matplotlibBlue] ($(Z3N) - (0,\SignThree)$) -- ($(Z3N) + (0,\SignThree)$);
\draw [{Bar[width=0.8mm]}-{Bar[width=0.8mm]}, draw=matplotlibBlue] ($(Z4N) - (0,\SignFour)$) -- ($(Z4N) + (0,\SignFour)$);

\draw [-{Latex[length=1.5mm]}, thick, black!60] ($(O)-(0,\AxOverlap)$) -- (0, 0.6);
\draw [-{Latex[length=1.5mm]}, thick, black!60] ($(O)-(\AxOverlap, 0)$) -- ($(O) + (5.3*\DeltaT, 0)$);
\end{tikzpicture}
}
		\put(41,34.5){\small $\bar{z}_{\text{I}}$}	
		\put(23,12.5){\small $\bar{z}_{\text{II}}$}	
		\put(58,12.5){\small $\bar{z}_{\text{III}}$}
		\put(26.5,10.2){\contractOne}	
		\put(43.8,32.4){\contractTwo}	
		\put(62.2,10.2){\contractThree}	
	\end{overpic}
	\caption{Subsystems in a distributed control scheme with physical couplings shown by dashed edges. They exchange information about nominal future trajectories of their local couplings, optionally also sensitivity information (here depicted as blue areas and intervals around the nominal values).}
	\label{fig: exchange nominal values}
\end{figure}
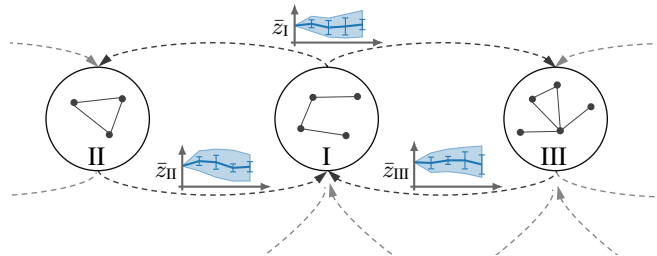
The main contribution of the paper is presented in Section~\ref{sec: sufficient conditions}, proving sufficient conditions under which the identification method successfully uncovers all attacked components even for nonlinear dynamics and nonlinear couplings of the subsystems. 
Remarkably, the proposed rigorous guarantees of the identification method can be applied for realistic nonlinear case studies, as we illustrate with experiments on the IEEE~30 bus power system in Section~\ref{sec: numerical experiments}.


\subsection{Distributed Control Setup}
A distributed system structure as in~\eqref{eq: local dynamics} suggests the application of distributed control methods, which typically scale much better than centralized approaches. For an overview of existing methods, in particular approaches in distributed model predictive control (MPC), we refer to the survey papers~\cite{Christofides2013Distributed,Scattolini2009Architectures}. In contrast to fully decentralized approaches, distributed control schemes are based on the exchange of some information between the subsystems. This typically allows to reduce the uncertainty in the mutual interference and can be employed to design local controllers that are robust towards unknown couplings of neighbored subsystems.
This idea is implemented in~\cite{Farina2012Distributed}, where the subsystems exchange corridors in which future coupling values are guaranteed to lie, and apply robust MPC controllers to approach the uncertainties. The concept is shown in Fig.~\ref{fig: exchange nominal values} and formally described in~\cite{Lucia2015Contract}, supplemented by conditions for stability guarantees.
In our previous work~\cite{Braun2020Hierarchical}, we applied the distributed robust MPC method from~\cite{Lucia2015Contract} to a nonlinear system of systems under attack, designing the local control inputs~$u_I$ to be robust towards uncertain coupling values~$z_{\mathcal{N}_I}$ of neighbors as well as potentially disturbed internal inputs~$a_I(u_I)$.
In this way, constraint satisfaction is achieved in each subsystem $I$, even if an attack disturbs $u_I$ or causes the neighbors' couplings $z_{\mathcal{N}_I}$ to deviate from the nominal values.
\begin{figure}[t]
	\centering
	\begin{overpic}[width=\columnwidth, tics=5]{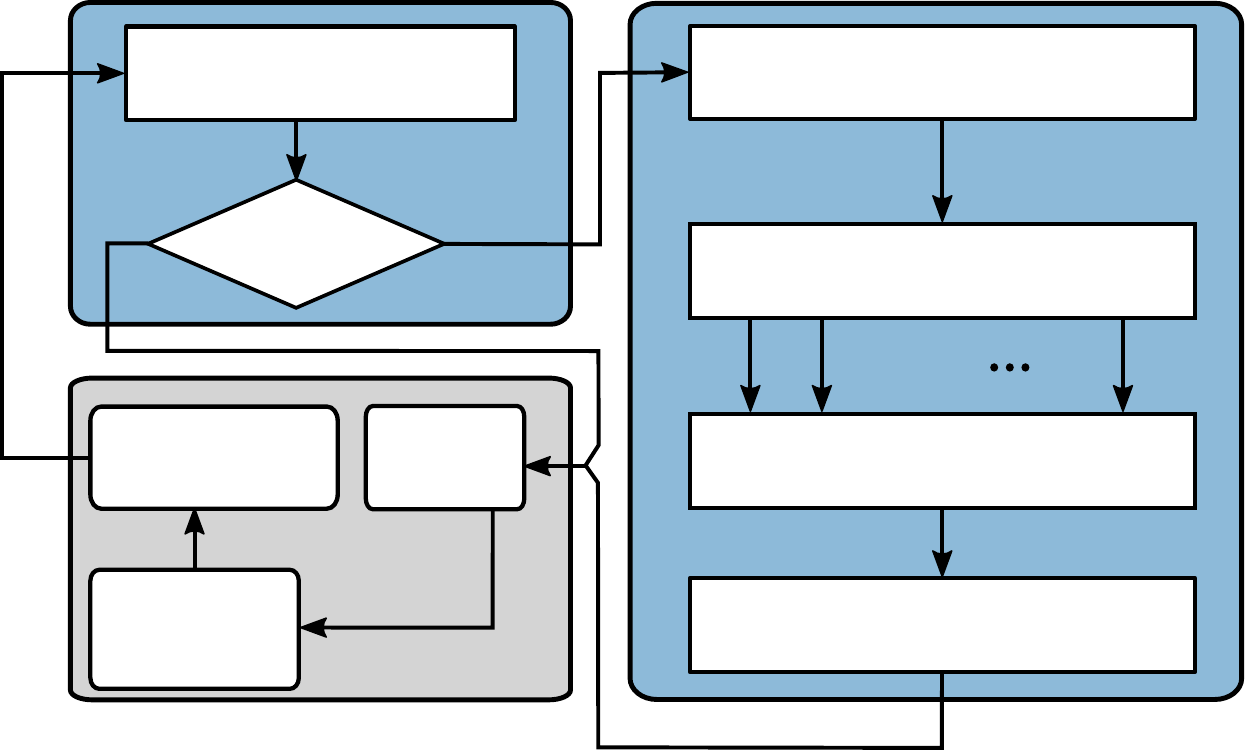} 
		\put(12,53){Attack Detection}
		\put(58,53){Attack Identification}	
		\footnotesize
		\put(16,39.9){\fontsize{7.5pt}{9pt}$\|\Delta z_I\| > \tau_\text{D}$}		
		\put(35,41.4){Yes}
		\put(8.5,41.4){No}
		\put(7.8,22.5){Measurements}
		\put(0.5,24){$\Delta z$}
		\put(31.1,22.5){DMPC}
		\put(37,17){$u$}
		\put(34.5,10.5){\includegraphics[width=0.4cm]{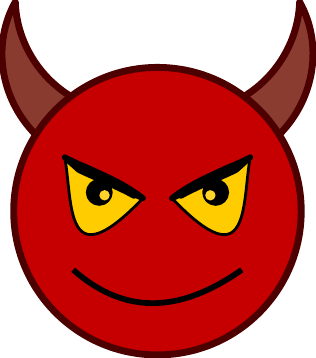}}
		\put(26.5,11){$a(u)$}
		\put(11,9){System}
		\put(56,37.5){Locally compute sensitivities}
		\put(66.5,30.5){\scriptsize$S_I^a, S_I^\mathcal{N}$}
		\put(64, 24){Set up and solve}
		\put(58, 20.5){identification problem~\eqref{opt: identification problem}}
		\put(76, 16.5){\scriptsize$\Delta a^\ast$}
		\put(61, 10.9){Identified attack set}
		\put(67, 7.4){$\text{supp}(a^\ast)$}
	\end{overpic}
	\caption{Outline of the hierarchical attack detection and identification, embedded in a distributed model predictive control (DMPC) loop and executed at each sampling time. Identification is based on exchanging locally computed sensitivity information and solving a central identification problem.}
	\label{fig: overview method}
\end{figure}
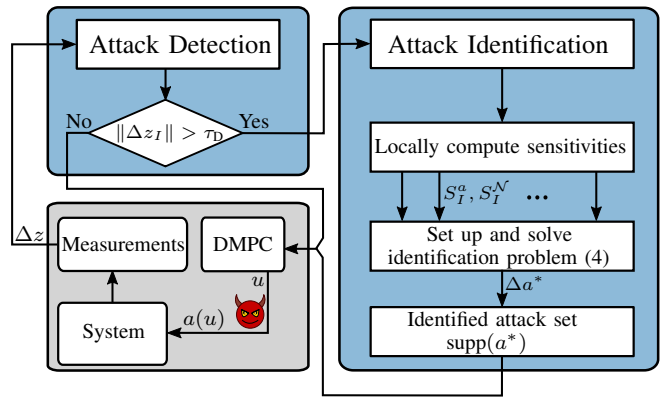
The identification method analyzed in this paper and illustrated in Fig.~\ref{fig: overview method} is applicable together with distributed closed-loop control schemes like the one in~\cite{Braun2020Hierarchical}. While it is mostly decoupled from the specific design of, e.g., the exchanged corridors, it requires the exchange of predicted nominal coupling values $\bar{z}_I$ as in Fig.~\ref{fig: exchange nominal values}. They indicate undisturbed reference values that the coupling variables $z_I$ will attain if no attack $a_I$ occurs in subsystem $I$ and the neighboring coupling variables $z_{\mathcal{N}_I}$ also behave according to their nominal values~$\bar{z}_{\mathcal{N}_I}$.
Closely following the notation in~\cite{Lucia2015Contract}, we denote by $\bar{z}_I(k|t)$ the nominal coupling value of subsystem $I$ for time $k$ calculated at time $t$. Similarly, $u_I(k|t)$ is the undisturbed input at time $k$ computed by the MPC scheme at time $t$ and $\bar{z}_{\mathcal{N}_I}(\cdot|t)$ is the function of nominal coupling values of neighboring subsystems on the prediction horizon, assumed to be discretized piecewise constant. The predicted nominal states $\bar{x}_I(k|t)$ and the nominal coupling values $\bar{z}_I(k|t)$ to be exchanged at time $t$ are computed as 
\begin{equation}
\label{eq: nominal values}
\begin{aligned}
\bar{x}_I(k|t) &\coloneqq f_I\left(x_I(k-1|t), u_I(k-1|t), \bar{z}_{\mathcal{N}_I}(\cdot|t)\right), \\
\bar{z}_I(k|t) &\coloneqq h_I\left(\bar{x}_I(k|t)\right),
\end{aligned}
\end{equation}
for $k = t+1, \dots, t+N$ with prediction horizon $N$.
After receiving the nominal trajectory $\bar{z}_J(\cdot|t)$ from each neighbor $J \in \mathcal{N}_I$ at time $t$, each subsystem $I$ combines its neighbors' nominal values for the next sampling time $t+1$ as
\begin{align*}
	\bar{z}_{\mathcal{N}_I}(k|t+1) \coloneqq \Pi_{J\in \mathcal{N}_I} \bar{z}_J(k|t).
\end{align*}
In order to obtain initial nominal coupling values $\bar{z}_I(k|0)$, we assume the system to be in steady state such that $h_I(x_I^0)$ for all $I \in \mathcal{P}$ provide suitable initial values. For a general procedure to obtain initial values we refer to~\cite{Braun2020Hierarchical}.

\section{HIERARCHICAL IDENTIFICATION METHOD}
\label{sec: identification method}
In accordance with relevant literature, such as~\cite{Pasqualetti2013Attack,Ananduta2020Resilient}, we distinguish between attack detection and identification as the problems to uncover the presence and location of an attack, respectively. Attack detectors typically monitor some system outputs and compare estimates with measurements to detect unexpected deviations that might indicate an attack~\cite{Dibaji2019Systems,Boem2018Plug}. For attack identification, we consider methods revealing the points of attack by means of the \emph{attack set}, which is defined in the following, similar to~\cite{Pasqualetti2013Attack}.
\begin{definition}[Attack Set]\label{def: attack set}
	Let $u \in \mathbb{U}$ be an undisturbed controller input and $a(u)$ the attacked input tampering with the dynamics according to~\eqref{eq: local dynamics}.  The \emph{attack set} $\text{supp}(a)$ of $a$ is defined as the set of all control indices which are affected by the attack, i.e., $\text{supp}(a) \coloneqq \{i: \left(a(u)\right)_i \neq u_i\} \subseteq \{1, \dots, d_u\}.$
\end{definition}

The blue highlighted fields in Fig.~\ref{fig: overview method} give an overview of the method for attack detection and identification that is presented in the following. It is embedded in a classical control loop with a distributed MPC controller and performed at each sampling time. Only if the detection scheme triggers an alarm, the identification method is executed. In the following, we consider one fixed sampling time and omit the time indices for the sake of brevity.
One step towards a hierarchical scheme (a detailed discussion follows) consists in monitoring the measurements of only the coupling variables~$z_I$ in each subsystem, instead of all global states $x$.
By definition, the nominal coupling values~$\bar{z}_I$ provide suitable estimates of the expected values in an undisturbed scenario. If in any subsystem $I$ the estimation error $\|\widetilde{z}_I - \bar{z}_I\|$ with measured coupling values $\widetilde{z}_I$ exceeds some detection threshold $\tau_\text{D}$, our detection method raises an alarm. Throughout this paper, we assume all coupling variables $z_I$ to be measurable without any measurement noise, i.e., $\widetilde{z}_I = z_I$. We further define the deviation $\widetilde{z}_I - \bar{z}_I$ from the nominal value as $\Delta z_I$. 

Since all subsystems are physically coupled, a significant deviation $\|\Delta z_I\| > \tau_\text{D}$ from the nominal values $\bar{z}_I$ in some subsystem $I$ may be caused by some internal attack $a_I$ in $I$, but may just as well result from an attack $a_J$ in some other subsystem $J\neq I$, the impact of which spreads through the network. 
The proposed attack identification is based on monitoring the deviations $\Delta z_I$ in the coupling values and figuring out at each time step $t$ in which subsystems the local inputs $u_I(t)$ are disturbed by some attack $a_I(u_I(t)) \neq u_I(t)$. For this purpose, we derive linear equations approximating the propagation of an attack through the network of subsystems.

According to the system dynamics~\eqref{eq: local dynamics}, the coupling variables $z_I= h_I \circ f_I(x_I, a_I(u_I), z_{\mathcal{N}_I})$ depend on $x_I, a_I(u_I)$ and $z_{\mathcal{N}_I}$, and we set $\zeta_I \coloneqq h_I \circ f_I$. The nominal coupling values are defined in~\eqref{eq: nominal values} such that \mbox{$\bar{z}_I = \zeta_I(x_I, u_I, \bar{z}_{\mathcal{N}_I})$.}
In order to analyze which deviations $\Delta z_I$ are caused by disturbances in $a_I(u_I)$ and $z_{\mathcal{N}_I}$, we compute a first-order Taylor approximation of $\zeta_I$ in $a_I(u_I)$ and $z_{\mathcal{N}_I}$ around the nominal value $(x_I, u_I, \bar{z}_{\mathcal{N}_I})$. Denoting the deviation $a_I(u_I) - u_I$ of the potentially disturbed input $a_I(u_I)$ from the undisturbed controller input $u_I$ by $\Delta a_I$, and the deviation $z_{\mathcal{N}_I} - \bar{z}_{\mathcal{N}_I}$ by $\Delta z_{\mathcal{N}_I}$, it holds by Taylor's theorem for $\Delta a_I, \Delta z_{\mathcal{N}_I} \rightarrow 0$:
\begin{equation}
\begin{aligned}
\label{eq: approximation attack propagation}
\Delta z_I 
&= \frac{\partial \zeta_I}{\partial a_I}(x_I, u_I, \bar{z}_{\mathcal{N}_I}) \Delta a_I \\
&+ \frac{\partial \zeta_I}{\partial z_{\mathcal{N}_I}}(x_I, u_I, \bar{z}_{\mathcal{N}_I})\Delta z_{\mathcal{N}_I} + R_I,
\end{aligned}
\end{equation}
where an estimation of the remainder term $R_I$ is given in Lemma~\ref{lem: estimation remainder term}.
The Jacobians $\frac{\partial \zeta_I}{\partial a_I}$ and $\frac{\partial \zeta_I}{\partial z_{\mathcal{N}_I}}$ evaluated at $(x_I, u_I, \bar{z}_{\mathcal{N}_I})$ are computed locally by each subsystem applying the chain rule on $\zeta_I = h_I \circ f_I$ and calculating
\begin{align*}
\frac{\partial \zeta_I}{\partial a_I}(x_I, u_I, \bar{z}_{\mathcal{N}_I}) = \frac{\partial h_I}{\partial a_I}(x_I) \frac{\partial f_I}{\partial a_I}(x_I, u_I, \bar{z}_{\mathcal{N}_I}).
\end{align*}
The Jacobian $\frac{\partial \zeta_I}{\partial z_{\mathcal{N}_I}}$ can be computed similarly. In the following, we denote these matrices by $S_I^a \coloneqq \frac{\partial \zeta_I}{\partial a_I}(x_I, u_I, \bar{z}_{\mathcal{N}_I})$ and \mbox{$S_I^\mathcal{N} \coloneqq \frac{\partial \zeta_I}{\partial z_{\mathcal{N}_I}}(x_I, u_I, \bar{z}_{\mathcal{N}_I})$.}
We assume that in the case of a detected attack all subsystems share locally evaluated sensitivity information by publishing $S_I^a$ and $S_I^\mathcal{N}$. Based on this data, equations~\eqref{eq: approximation attack propagation} for each subsystem $I$ omitting the remainder term $R_I$ provide a linear approximation of the attack propagation through the network. For attack identification, we compute an attack with the sparsest possible attack set that explains the observed deviations $\Delta z_I$ by satisfying the linearized propagation equations. To this end, the following sparse signal recovery problem is solved:
\begin{equation}
\label{opt: identification problem}
\begin{aligned}
&\min_{\Delta a} && \|\Delta a\|_0 \\
&\text{ s.t.}  &&S^a_I \Delta a_I = \Delta z_I -  S^{\mathcal{N}}_I \Delta z_{\mathcal{N}_I}  ~~\forall I \in \mathcal{P}.
\end{aligned}
\end{equation}
Here, $\|\Delta a\|_0$ denotes the $\ell_0$-``norm'' of $\Delta a$, counting the nonzero elements in $\Delta a$. For the corresponding attack $a$ with $\Delta a = a(u)- u$ it thus holds $|\text{supp}(a)| = \|\Delta a\|_0$ for the attack set $\text{supp}(a)$ as in Definition~\ref{def: attack set}.
Hence, an optimal solution $\Delta a^\ast$ of~\eqref{opt: identification problem} corresponds to an attack with the sparsest attack set among all attacks that fulfill the linear approximation of the attack propagation. Searching for a sparsest possible attack is a common approach for attack identification, see for example~\cite{Pasqualetti2013Attack,Pan2015Online,Liu2014Detecting}. It can be justified by the fact that attackers typically have restricted resources, so they can only disturb in a limited number of nodes~\cite{Liu2014Detecting}.
Since solving the $\ell_0$-minimization problem~\eqref{opt: identification problem} involves a mixed-integer program and is thus NP-hard, the $\ell_0$-``norm'' is commonly relaxed by the $\ell_1$-norm, which turns problem~\eqref{opt: identification problem} into a linear optimization problem~\cite{Braun2020Hierarchical,Pan2015Online,Candes2005Decoding}. In this paper, however, we focus on provable statements with the linearized attack propagation in the constraints and do not introduce another approximation error but \mbox{stick to the $\ell_0$-``norm''.}

Due to the fact that the identification problem~\eqref{opt: identification problem} involves measured coupling deviations $\Delta z_I$ and sensitivity information $S_I^a$, $S_I^\mathcal{N}$ for all subsystems $I$, it is not a distributed identification method. 
But it is also not a classical centralized method since no information about the local dynamics~$f_I$, coupling functions~$h_I$ nor individual cost functions is needed. Assuming that the subsystems agree to provide the required sensitivities and measurements to some superior instance that solves the identification problem, it can be considered hierarchical.
Additionally, it requires only the couplings but not all states to be measured. Since problem~\eqref{opt: identification problem} contains $d_u$ optimization variables, it can be expected to scale significantly better than a fully centralized nonlinear method involving $d_x + d_u$ variables affecting the global dynamics.

\section{SUFFICIENT CONDITIONS FOR GUARANTEED ATTACK IDENTIFICATION}
\label{sec: sufficient conditions}
We consider some fixed sampling time $t$ at which an unknown attack $\widehat{a}$ disturbs the controller input $u(t)$ by \mbox{$\Delta \widehat{a} = \widehat{a}(u(t)) - u(t)$} and causes deviations $\Delta \widehat{z}$ in the coupling variables.
Only for the special case of $\zeta_I$ being linear for all $I$, the actually occurring attack $\Delta \widehat{a}$ satisfies the first-order approximation of the attack propagation and is a feasible solution of the identification problem~\eqref{opt: identification problem}. Even for systems with linear dynamics $\dot{x} = Ax + Ba(u)$ and linear coupling equations $z_I = H_I x_I$, however, the resulting functions $\zeta_I$ can be nonlinear since the solution of a linear ODE is in general nonlinear. 
In this section, we consider nonlinear functions~$\zeta_I$ and derive suitable assumptions under which a solution of the identification problem~\eqref{opt: identification problem} identifies an attack $\Delta a^\ast$ that is close to the actual attack $\Delta \widehat{a}$ in an appropriate manner. Instead of bounding the error \mbox{$\|\Delta a^\ast - \Delta \widehat{a}\|$} with the $\ell_1$- or $\ell_2$-norm, we are interested in results stating that the actual, unknown attack set $\text{supp}(\widehat{a})$ (or some superset) is correctly identified. The two main results of this paper, given in Theorems~\ref{thm: superset identification} and~\ref{thm: correct identification}, provide statements of this kind.

In order to analyze the approximation error of the linearized attack propagation constituting the constraints of the identification problem~\eqref{opt: identification problem}, we consider the remainder term~$R_I$ in~\eqref{eq: approximation attack propagation} and derive an upper bound for $\|R_I\|_2$ in Lemma~\ref{lem: estimation remainder term}. For this purpose, we make use of the multi-index notation for derivatives of multivariate functions, see, e.g.,~\cite{Forster2011Analysis}. For a multi-index $\alpha = (\alpha_1, \alpha_2, \dots, \alpha_n) \in \mathbb{N}^n$, a real vector $x=(x_1, x_2, \dots, x_n) \in \mathbb{R}^n$ and some smooth function $g: \mathbb{R}^n \rightarrow \mathbb{R}^m$ we define
\begin{align*}
&|\alpha| \coloneqq \alpha_1 + \alpha_2 + \dots + \alpha_n,  
&&\alpha! \coloneqq \alpha_1!\alpha_2! \dots \alpha_n!, \\
&x^\alpha \coloneqq x_1^{\alpha_1}x_2^{\alpha_2} \dots x_n^{\alpha_n} ~\text{ and}
&&\partial^\alpha g \coloneqq \frac{\partial^{|\alpha|}g}{\partial x_1^{\alpha_1}\partial x_2^{\alpha_2}\dots \partial x_n^{\alpha_n}}.
\end{align*}

\begin{lemma}[Estimation of Remainder Term]
	\label{lem: estimation remainder term}
	Let for all $I$ the function $\zeta_I = h_I \circ f_I$ be twice continuously differentiable. We assume that at some fixed $x_I$ the maximum second-order partial derivative $K_I \coloneqq \max_{|\alpha|=2}\left\|\partial^\alpha \zeta_I(x_I, \cdot, \cdot)\right\|_2$ exists and is finite, and define $K \coloneqq \max_I K_I$.
	For the remainder term $R_I$ of the first-order Taylor approximation of $\zeta_I$ it holds 
	\begin{align*}
	\|R_I\|_2 \leq \frac{K_I}{2} \big(\|\Delta a_I\|_1 + \|\Delta z_{\mathcal{N}_I}\|_1\big)^2.
	\end{align*}
	For the total remainder term $R = \left(R_I\right)_{I \in \mathcal{P}}$ it holds 
	\begin{align*}
	\|R\|_2 \leq \frac{K}{2} \big(\|\Delta a\|_1 + M\|\Delta z\|_1\big)^2,
	\end{align*}
	with $M \coloneqq \max_I |\mathcal{N}_I|$ denoting the maximum degree in the network where each subsystem $I$ constitutes one node.
\end{lemma}
\begin{proof}
	According to Theorem 2 in §7 of~\cite{Forster2011Analysis}, it holds for the remainder term $R_I$
	\begin{align*}
	R_I = \sum_{|\alpha|=2} \partial^\alpha \zeta_I(x_I, \xi^a_{I}, \xi^{z_{\mathcal{N}}}_I)\frac{1}{\alpha!} \begin{pmatrix}\Delta a_I\\\Delta z_{\mathcal{N}_I}\end{pmatrix}^\alpha,
	\end{align*}
	with \mbox{$\xi^a_{I} = u_I + c_I^a \Delta a_I$}, \mbox{$\xi^{z_{\mathcal{N}}}_I = \bar{z}_{\mathcal{N}_I} + c_I^\mathcal{N} \Delta z_{\mathcal{N}_I}$} intermediate points for some $c_I^a, c_I^\mathcal{N} \in (0,1)$.
	Using the triangle inequality and the definition of $K_I$, we obtain
	\begin{align*}
	\|R_I\|_2 
	&\leq \sum_{|\alpha|=2} \bigg\|\partial^\alpha \zeta_I(x_I, \xi^a_{I}, \xi^{z_{\mathcal{N}}}_I)\frac{1}{\alpha!} \underbrace{\begin{pmatrix}\Delta a_I\\\Delta z_{\mathcal{N}_I}\end{pmatrix}^\alpha}_{\in \mathbb{R}} \bigg\|_2 \\
	&=\sum_{|\alpha|=2} \big\|\partial^\alpha \zeta_I(x_I, \xi^a_{I}, \xi^{z_{\mathcal{N}}}_I) \big\|_2 \frac{1}{\alpha!} \left|\begin{pmatrix}\Delta a_I\\\Delta z_{\mathcal{N}_I}\end{pmatrix}^\alpha\right| \\
	&\leq K_I \sum_{|\alpha|=2}\frac{1}{\alpha!} \left|\begin{pmatrix}\Delta a_I\\\Delta z_{\mathcal{N}_I}\end{pmatrix}\right|^\alpha \\
	&= \frac{K_I}{2} \big(\|\Delta a_I\|_1 + \|\Delta z_{\mathcal{N}_I}\|_1\big)^2.
	\end{align*}
	The last equality holds due to the multinomial theorem for $k=2$, which states the equality
	\mbox{$(x_1 + x_2 + \dots + x_n)^k$} $= \sum_{|\alpha|=k} \frac{k!}{\alpha!} x^\alpha$
	and can be proven using the binomial theorem and induction on $n$.
	It remains to derive an upper bound for the total remainder term $R = \left(R_I\right)_{I \in \mathcal{P}}$.
	We estimate
	\begin{align*}
	\|R\|_2 
	&\leq \sum_{I\in \mathcal{P}}\|R_I\|_2
	\leq \sum_{I\in \mathcal{P}}\frac{K_I}{2} \big(\|\Delta a_I\|_1 + \|\Delta z_{\mathcal{N}_I}\|_1\big)^2\\
	&\leq \frac{K}{2} \sum_{I\in \mathcal{P}} \big(\|\Delta a_I\|_1 + \|\Delta z_{\mathcal{N}_I}\|_1\big)^2\\
	&\leq \frac{K}{2} \left(\left\|\begin{pmatrix}\Delta a_1 \\ \vdots \\ \Delta a_{|\mathcal{P}|}\end{pmatrix}\right\|_1 + \left\|\begin{pmatrix}\Delta z_{\mathcal{N}_1} \\ \vdots \\ \Delta z_{\mathcal{N}_{|\mathcal{P}|}}\end{pmatrix}\right\|_1\right)^2,
	\end{align*}
	where the last inequality also follows from the multinomial theorem.
	For the first vector in the last line it holds \mbox{$\begin{pmatrix}\Delta a_1, \dots, \Delta a_{|\mathcal{P}|}\end{pmatrix}^\text{T} = \Delta a^\text{T}$,} but for the second vector it holds in general $\left\|\begin{pmatrix}\Delta z_{\mathcal{N}_1}, \dots, \Delta z_{\mathcal{N}_{|\mathcal{P}|}}\end{pmatrix}\right\|_1 \neq \|\Delta z\|_1$ since each vector $\Delta z_{I}$ appears $|\mathcal{N}_I|$ many times. With $M$ denoting the maximum degree in the network, it holds 
	\begin{align*}
	\left\|\begin{pmatrix}\Delta z_{\mathcal{N}_1} \\ \vdots \\ \Delta z_{\mathcal{N}_{|\mathcal{P}|}}\end{pmatrix}\right\|_1 
	&= \sum_I |\mathcal{N}_I| \|\Delta z_I\|_1 
	\leq M \|\Delta z\|_1.
	\end{align*}
	In total, we obtain
	\begin{align*}
	\|R\|_2 \leq \frac{K}{2} \big(\|\Delta a\|_1 + M\|\Delta z\|_1\big)^2.
	\end{align*}
\end{proof}
\noindent
Using this upper bound on the remainder term, we next derive an $\varepsilon$-$\delta$-criterion that specifies a condition under which the computed solution $\Delta a^\ast$ of~\eqref{opt: identification problem} is in an \mbox{$\varepsilon$-neighborhood} around the actual attack $\Delta \widehat{a}$. 
For the sake of clarity, we express the linear constraints of problem~\eqref{opt: identification problem} in the form $S \Delta a = b$ with $S = \text{diag}\left((S_I^a)_{I\in \mathcal{P}}\right) \in \mathbb{R}^{d_z \times d_u}$ and \mbox{$b=\left(\Delta z_I - S_I^\mathcal{N} \Delta z_{\mathcal{N}_I}\right)_{I\in \mathcal{P}}$.} The smallest singular value of~$S$ is denoted as $\sigma_{\min}$.
\begin{lemma}[$\varepsilon$-$\delta$-Criterion]
	\label{lem: eps-delta-criterion}
	We assume that $\sigma_{\min} > 0$ and $d_z \geq d_u$ holds for $d_z = \sum_{I\in \mathcal{P}} d_{z_I}$ denoting the total number of coupling variables. Let $\varepsilon > 0$ be given and denote by $\Delta a^\ast$ a feasible solution of the identification problem~\eqref{opt: identification problem}. Defining~$\delta$ as $\delta \coloneqq \sqrt{\frac{2\varepsilon\sigma_{\min}}{K}}$ it holds:
	
	If $\left(\|\Delta \widehat{a}\|_1 + M\|\Delta \widehat{z}\|_1\right) \leq \delta$, then 
	$\left\|\Delta \widehat{a} - \Delta a^\ast\right\|_2 \leq \varepsilon$.
\end{lemma}
\begin{proof}
	The main idea of the proof is to make use of the linearity of the constraints in~\eqref{opt: identification problem} to bound the distance between $\Delta \widehat{a}$ and $\Delta a^\ast$. A feasible solution $\Delta a^\ast$ clearly satisfies the constraints such that 
	$b - S \Delta a^\ast = 0$ holds.
	For the actual attack $\Delta \widehat{a}$ it holds $b - S \Delta \widehat{a} = R$ with $R$ being the remainder term from the Taylor expansion. Subtracting these equations, we obtain
	$$
	\left\|R\right\|_2 = 
	\left\|S (\Delta \widehat{a} - \Delta a^\ast)\right\|_2.
	$$
	Since $d_z \geq d_u$, a lower bound of this expression is given by
	$$
	\left\|R\right\|_2 \geq \sigma_{\min} \left\|\Delta \widehat{a} - \Delta a^\ast\right\|_2,
	$$
	with $\sigma_{\min} > 0$ denoting the smallest singular value of $S$.
	Using the upper bound of the remainder term from Lemma~\ref{lem: estimation remainder term} and the definition of $\delta$, it follows
	\begin{align*}
	\left\| \Delta \widehat{a} - \Delta a^\ast \right\|_2 &\leq \frac{\left\|R\right\|_2}{\sigma_{\min}} \leq \frac{K}{2\sigma_{\min}}\left(\|\Delta \widehat{a}\|_1 + M\|\Delta \widehat{z}\|_1\right)^2 \\
	&\leq \frac{K}{2\sigma_{\min}} \delta^2 = \varepsilon.
	\end{align*}
\end{proof}
We would like to derive conditions under which the attack sets $\text{supp}(\widehat{a})$ and $\text{supp}(a^\ast)$ are similar rather than the attack vectors $\Delta \widehat{a}$ and $\Delta a^\ast$ themselves. In other words, we are interested in a specific $\varepsilon$ such that Lemma~\ref{lem: eps-delta-criterion} implies that both~$\widehat{a}$ and $a^\ast$ have the same attack set. First, we state a slightly weaker result, implying that under the indicated conditions all attacked inputs are identified by the computed solution, but possibly \mbox{also some benign components are suspected.}

\begin{theorem}[Correct Superset-Identification]
	\label{thm: superset identification}
	Let again $\sigma_{\min}> 0$, $d_z \geq d_u$ hold, let $M$ denote the maximum degree and $\Delta a^\ast$ a feasible solution of the identification problem~\eqref{opt: identification problem}. 
	Let $\varepsilon > 0$ be such that $\varepsilon < \min_{i \in \text{supp}(\widehat{a})} |(\Delta \widehat{a})_i|$ and choose~$\delta$ accordingly as in Lemma~\ref{lem: eps-delta-criterion}. If $\left(\|\Delta \widehat{a}\|_1 + M\|\Delta \widehat{z}\|_1\right) \leq \delta$ holds, then for the attack sets we have
	\begin{align*}
	\text{supp}(a^\ast) \supseteq \text{supp}(\widehat{a}).
	\end{align*}
\end{theorem}
\begin{proof}
	From Lemma~\ref{lem: eps-delta-criterion} it follows that \mbox{$\|\Delta \widehat{a} - \Delta a^\ast\|_2 \leq \varepsilon.$}
	We assume for contradiction that $\text{supp}(a^\ast) \nsupseteq \text{supp}(\widehat{a})$. Hence, there is some index $i \in \text{supp}(\widehat{a})$ but $i \notin \text{supp}(a^\ast)$, i.e., $(\Delta \widehat{a})_i \neq 0$ and $(\Delta a^\ast)_i = 0$. This implies
	\begin{align*}
	\|\Delta \widehat{a} - \Delta a^\ast\|_2 
	&\geq |(\Delta \widehat{a})_i - (\Delta a^\ast)_i| = |(\Delta \widehat{a})_i| \\
	&\geq \min_{i \in \text{supp}(\widehat{a})} |(\Delta \widehat{a})_i| > \varepsilon,
	\end{align*}
	which contradicts the result following from Lemma~\ref{lem: eps-delta-criterion}.\\
\end{proof}
Theorem~\ref{thm: superset identification} guarantees, under certain assumptions, that a solution of the identification problem identifies all attacked inputs, but possibly also some undisturbed inputs. In the numerical experiments in Section~\ref{sec: numerical experiments} we will analyze how large the discrepancy is on average for randomly generated attacks. To achieve equality of the attack sets $\text{supp}(\widehat{a})$ and $\text{supp}(a^\ast)$ and thus guarantee that $\Delta a^\ast$ correctly identifies all attackers but no more, some modifications are necessary. Due to the nonlinearity of $\zeta_I$ the approximation of the attack propagation is not exact and the actual attack~$\Delta \widehat{a}$ in general does not have to be a feasible solution of~\eqref{opt: identification problem}. To resolve this, we consider a relaxed version of the identification problem:
\begin{equation}
\label{opt: identification problem relaxed}
\begin{aligned}
&\min_{\Delta a} &&\|\Delta a\|_0 \\
&\text{ s.t.}  &&\left\|b - S\Delta a\right\|_2 \leq \frac{\varepsilon}{2}\sigma_{\min},
\end{aligned}
\end{equation}
where again $\varepsilon < \min_{i \in \text{supp}(\widehat{a})} |(\Delta \widehat{a})_i|$ and $\sigma_{\min}$ is the smallest singular value of the sensitivity matrix $S$. Slightly modifying the definition of $\delta$ by a constant factor and requiring $\Delta a^\ast$ to be a global solution, we obtain the following stronger result:
\begin{theorem}[Exact Identification]
	\label{thm: correct identification}
	Assume $\sigma_{\min}>0$, $d_z \geq d_u$ and let $\Delta a^\ast$ be a globally optimal solution of the relaxed problem~\eqref{opt: identification problem relaxed}. For \mbox{$\varepsilon < \min_{i \in \text{supp}(\widehat{a})} |(\Delta \widehat{a})_i|$,} we define $\tilde{\delta} \coloneqq \sqrt{\frac{\varepsilon\sigma_{\min}}{K}}$. If the actual attack $\Delta \widehat{a}$ satisfies $\left(\|\Delta \widehat{a}\|_1 + M\|\Delta \widehat{z}\|_1\right) \leq \tilde{\delta}$, then it holds 
	$$
	\text{supp}(a^\ast) = \text{supp}(\widehat{a}).
	$$
\end{theorem}
\begin{proof}
	As a first step we show that the proof of Lemma~\ref{lem: eps-delta-criterion} works similarly for the relaxed identification problem~\eqref{opt: identification problem relaxed} and the adapted $\tilde{\delta}$. The expression $b - S \Delta a^\ast$ is no longer zero and we define the corresponding residual as $R^\ast \coloneqq b - S \Delta a^\ast$ with \mbox{$\|R^\ast\|_2 \leq \frac{\varepsilon}{2} \sigma_{\min}$} due to feasibility. Similar to the proof of Lemma~\ref{lem: eps-delta-criterion} we estimate
	\begin{align*}
	\left\|\Delta \widehat{a} - \Delta a^\ast\right\|_2 
	&\leq 
	\frac{\left\|R - R^\ast\right\|_2}{\sigma_{\min}} 
	\leq 
	\frac{\left\|R\| + \|R^\ast\right\|_2}{\sigma_{\min}}\\
	&\leq
	\frac{1}{\sigma_{\min}} \left(\frac{K\widetilde{\delta}^2 + \varepsilon \sigma_{\min}}{2}\right) =
	\varepsilon.
	\end{align*}
	We have thus shown a similar result as in Lemma~\ref{lem: eps-delta-criterion} and a proof analogously to the one of Theorem~\ref{thm: superset identification} follows accordingly. Therefore, we obtain
	\begin{align}
	\label{eq: support inclusion one way}
	\text{supp}(a^\ast) \supseteq \text{supp}(\widehat{a})
	\end{align}
	for $\Delta a^\ast$ being a solution of the relaxed identification problem~\eqref{opt: identification problem relaxed}. It remains to show that $\text{supp}(a^\ast) \subseteq \text{supp}(\widehat{a})$.
	To this end, we note that the actual attack $\Delta \widehat{a}$ is a feasible solution of the relaxed problem~\eqref{opt: identification problem relaxed} since 
	\begin{align*}
	\left\|b - S \Delta \widehat{a}\right\|_2 
	&\leq \frac{K}{2}\big(\|\Delta \widehat{a}\|_1 + M\|\Delta \widehat{z}\|_1\big)^2 \leq \frac{K}{2} \tilde{\delta}^2 
	= \frac{\varepsilon}{2}\sigma_{\min}.
	\end{align*}
	Since both $\Delta \widehat{a}$ and $\Delta a^\ast$ are feasible solutions of~\eqref{opt: identification problem relaxed} and $\Delta a^\ast$ is globally optimal, it holds $\|\Delta a^\ast\|_0 \leq \|\Delta \widehat{a}\|_0$. Together with~\eqref{eq: support inclusion one way} (implying $\|\Delta a^\ast\|_0 \geq \|\Delta \widehat{a}\|_0$) it follows \mbox{$\|\Delta a^\ast\|_0 = \|\Delta \widehat{a}\|_0$.} Since $\text{supp}(a^\ast) \supseteq \text{supp}(\widehat{a})$, this implies $\text{supp}(a^\ast) = \text{supp}(\widehat{a})$.\\
\end{proof}
\begin{remark}\label{rem: reduction sensitivities}
	The assumptions $\sigma_{\min} > 0$ and $d_z\geq d_u$ can be replaced without loss of generality by assuming that the subsystems do not transmit the Jacobians \mbox{$S_I^a \in \mathbb{R}^{d_{z_I} \times d_{u_I}}$,} but instead remove dependent columns and publish submatrices $\widetilde{S}_I^a \in \mathbb{R}^{d_{z_I} \times r_I}$ of full rank $r_I \leq \min\{d_{z_I}, d_{u_I}\}$. So they omit redundant information which only further reduces the number of variables in problems~\eqref{opt: identification problem} and~\eqref{opt: identification problem relaxed}. It yields a total sensitivity matrix~$\widetilde{S} = \text{diag}\left((\widetilde{S}_I^a)_{I\in \mathcal{P}}\right)$ of \mbox{size $d_z \times r$ with} $r= \sum_I r_I \leq d_z$, so the proof of Lemma~\ref{lem: eps-delta-criterion} follows as above.
\end{remark}

\section{ATTACK IDENTIFICATION \\IN POWER SYSTEMS}
\label{sec: numerical experiments}
In order to evaluate the identification method from Section~\ref{sec: identification method}, we consider the problem of identifying faulty buses in power systems. For randomly generated attack scenarios, we analyze the ratio of correctly identified (supersets of the) attack sets and the proportion of samples where the sufficient conditions of Theorems~\ref{thm: superset identification} and~\ref{thm: correct identification} are satisfied, respectively. This allows us to assess not only the effectiveness of the identification method for nonlinear systems, but also the relevance of our main statements in Theorems~\ref{thm: superset identification} and~\ref{thm: correct identification}.

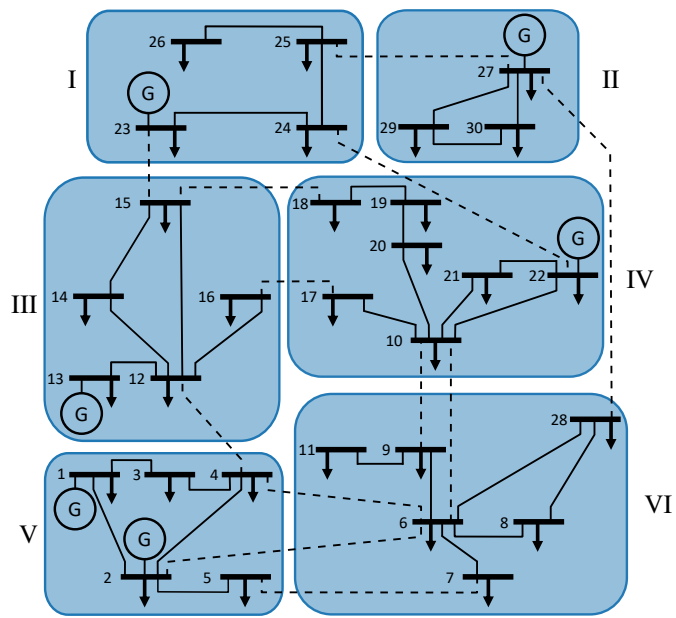
\begin{figure}[htb]
	\centering
	\begin{overpic}[tics=5, width=0.9\columnwidth]{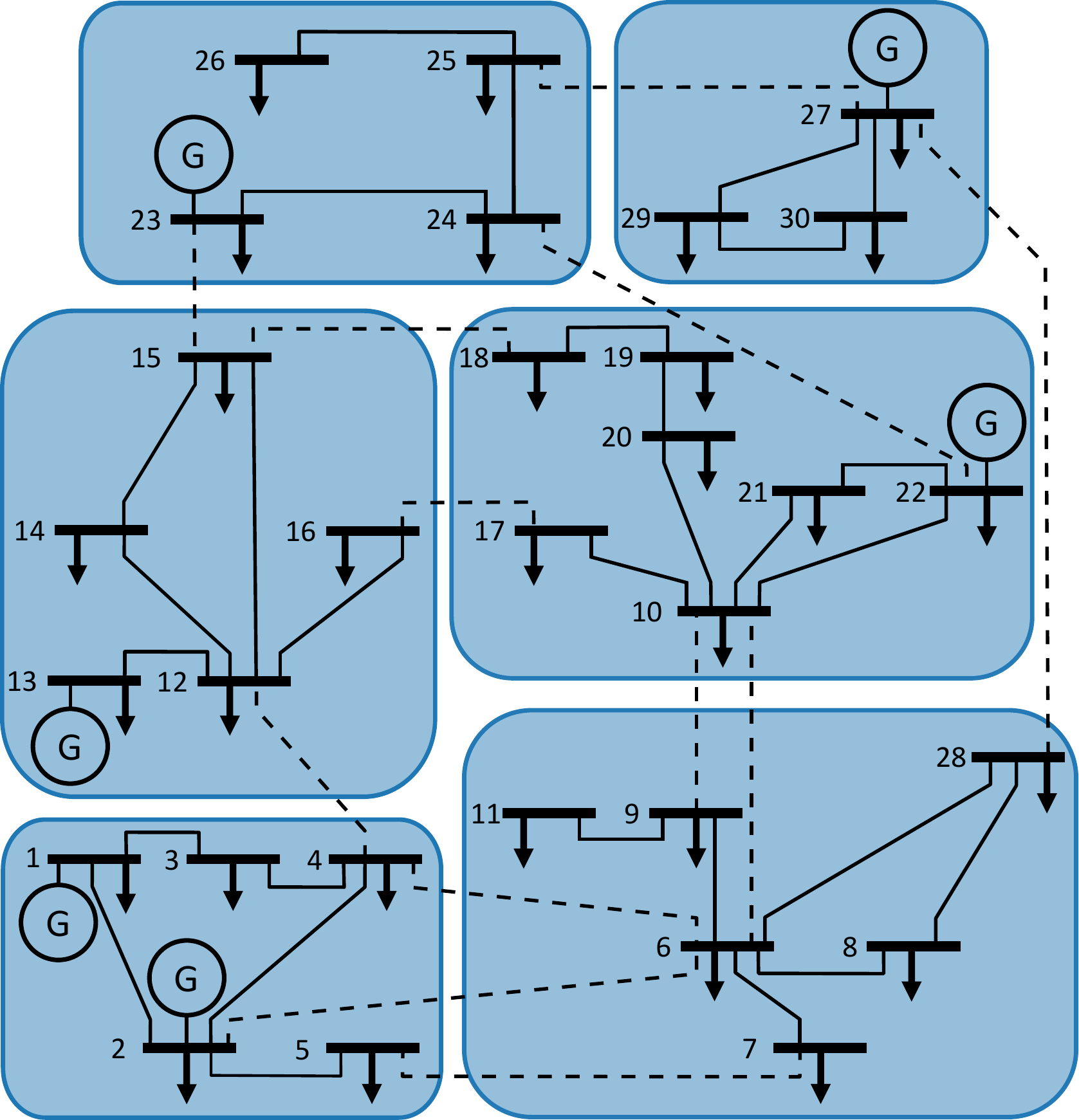}
		\put(4,87){\normalsize I}
		\put(92,87){\normalsize II}
		\put(-5.3,50){\normalsize III}		
		\put(96,54){\normalsize IV}			
		\put(-4,12){\normalsize V}	
		\put(99,17){\normalsize VI}
	\end{overpic}	
	\caption{Schematic of the IEEE~30 bus system partitioned into six subsystems I-VI. Physical couplings through transmission lines between two subsystems are depicted as dashed lines.}
	\label{fig: IEEE 30 bus system}
\end{figure}
We consider the IEEE~30 bus system shown in Fig.~\ref{fig: IEEE 30 bus system}, which consists of 30 buses all of which we assume to be connected to synchronous machines. The dynamics of the machine in bus $i$ with phase angle $\theta_i$ can thus be modeled by the so-called swing equation, see~\cite{Kundur1994Power}:
\begin{align*}
m_i \ddot{\theta}_i + d_i \dot{\theta}_i = u_i - \sum_{j \in N_i} P_{ij},
\end{align*}
where $m_i$ and $d_i$ denote inertia and damping constants, $u_i$ is the power infeed at bus $i$ and $P_{ij}$ describes the active power flow from bus $i$ to some bus $j$ in its neighborhood $N_i$. For the six generators buses 1, 2, 13, 22, 23 and 27, the dynamic coefficients $m_i$ and $d_i$ are taken based on the values in~\cite{DeTuglie2008coherency} and the conversion rules in~\cite{Kundur1994Power}. For the remaining load buses, arbitrary coefficients in a realistic range are chosen. 
If the dynamics of power lines are neglected, the power flow $P_{ij}$ between neighbored buses $i$ and $j$ can be modeled by 
\begin{align*}
P_{ij} = |V_i||V_j|b_{ij}\sin(\theta_i - \theta_j),
\end{align*}
with $|V_i|$ denoting the voltage magnitude at bus $i$, and $b_{ij}$ the susceptance of the transmission line between buses $i$ and $j$. Realistic parameter values and initial values for $\theta$ are taken from a simulation of the corresponding power system in Matpower~\cite{Zimmerman2010MATPOWER}. All parameters are chosen in a per-unit (p.u.) system with a 200\si{kV} base and a nominal frequency of 60\si{Hz}.
We consider constant loads at the six buses 3, 7, 14, 19, 26 and 30 and assume that the power infeeds at the remaining load and generator buses can be controlled through $u_i$ with $u_i \in [-0.4, 0]$ p.u.\ at all load buses and $u_i \in [-0.4, 0.9]$ p.u.\ at all generator buses. For frequency control of the system, we consider the following optimal control problem with states $\theta_i$, $\omega_i \coloneqq \dot{\theta_i}$ for $i=1,\dots,30$, and parameters $k_{ij} \coloneqq |V_i||V_j|b_{ij}$:
\begin{align}
&\min_{\theta, \omega, u} &&\|\omega\|_2^2 \notag\\
&\text{ s.t. }
&& \dot{\theta_i} = \omega_i, \label{opt: optimal frequency control}\\
&&& \dot{\omega_i} = \frac{1}{m_i}\bigg(u_i - d_i\omega_i - \sum_{j \in N_i} k_{ij} \sin\left(\theta_i - \theta_j\right)\bigg).\notag
\end{align}
An optimal solution of problem~\eqref{opt: optimal frequency control} minimizes the deviation~$\omega$ from the nominal frequency while obeying the power flow and machine dynamics. In our experiments, we consider a time horizon of $10$\si{s}, discretized with time steps of length $\Delta t = 0.1\si{s}$, and solve problem~\eqref{opt: optimal frequency control} in a distributed receding-horizon fashion applying the robust MPC scheme from~\cite{Lucia2015Contract}. It is implemented based on the do-mpc environment for multi-stage MPC~\cite{Lucia2017Rapid}, applying the NLP solver Ipopt~\cite{Wachter2006Implementation} and CasADi for automatic differentiation and optimization~\cite{Andersson2019CasADi}.
In the distributed scheme, one local MPC controller is used for each of the six subsystems indicated in Fig.~\ref{fig: IEEE 30 bus system}, which are interconnected through transmission lines drawn as dashed lines. To model the resulting physical coupling, in each subsystem those phase angles $\theta_i$ are defined as coupling variables which are incident to at least one dashed edge. In subsystem V, for example, the coupling variables \mbox{$z_\text{V}=(\theta_2, \theta_4, \theta_5)$} influence the neighbored \mbox{subsystems III and VI.} The coupling variables are assumed to be parametrized piecewise constant in the numerical integration scheme.
The partition of the IEEE~30 bus system into the indicated six subsystems yields a total number of $d_z=18$ coupling variables, which is significantly less than $d_x=60$ states and underlines again the reduced complexity of the proposed procedure, which does not require global measurements of all states nor knowledge of the local dynamics. 
As there are $d_u=30\nleq d_z$ input variables, we assume that the subsystems publish full-rank submatrices $\widetilde{S}_I^a$ instead of the original sensitivity \mbox{matrices $S_I^a$ as described in Remark~\ref{rem: reduction sensitivities}.}

To evaluate the identification method from Section~\ref{sec: identification method} and the strength of the sufficient conditions of Theorems~\ref{thm: superset identification} and~\ref{thm: correct identification}, we carry out two test series \texttt{attack\_1} and \texttt{attack\_3}. In both, the system is exposed to a new, randomly generated attack at each of the 100 time steps in $[0,10]\si{s}$ and the proposed detection and identification method depicted in Fig.~\ref{fig: overview method} is applied at each sampling time. In \texttt{attack\_1}, at each time step $t$, one attacked node~$i$ and a disturbed input value $a_i(u_i(t)) \neq u_i(t)$ are chosen uniformly at random. For the remaining nodes $j\neq i$, the undisturbed controller input $a_j(u_j(t)) = u_j(t)$ is applied to the system. In \texttt{attack\_3} three random nodes per time step are attacked. An attack is detected at time step $t$ if $\|\Delta z_I(t)\|_\infty > \tau_\text{D}$ for some $I$ with detection threshold $\tau_\text{D} \coloneqq 10^{-5}$. If this is the case, the sensitivity matrices $\widetilde{S}_I^a$, $S_I^\mathcal{N}$ are locally evaluated by applying automatic differentiation with CasADi to the local integrator schemes, representing the functions~$f_I$ and $h_I$ in equations~\eqref{eq: local dynamics}. Normalizing the columns of the matrices $\widetilde{S}_I^a$ and aggregating all sensitivity information, the identification problems~\eqref{opt: identification problem} and~\eqref{opt: identification problem relaxed} are set up and solutions $\Delta a^\ast_\eqref{opt: identification problem}$ and $\Delta a^\ast_\eqref{opt: identification problem relaxed}$ are computed with Bonmin, respectively~\cite{Bonami2008algorithmic}. The identified attack sets $\text{supp}(a^\ast_\eqref{opt: identification problem})$, $\text{supp}(a^\ast_\eqref{opt: identification problem relaxed})$ contain those indices $i$, for which $|\Delta a^\ast_\eqref{opt: identification problem}|_i > \varepsilon_\text{I}$ resp.\ $|\Delta a^\ast_\eqref{opt: identification problem relaxed}|_i > \varepsilon_\text{I}$ holds with identification threshold $\varepsilon_\text{I} \coloneqq 10^{-5}$. 

Among the 100 time steps with random attack sets of cardinality 1 in \texttt{attack\_1}, the detection gives an alarm at 79 sampling times. This seemingly low rate is due to the fact that only one input $u_i$ is modified by some random disturbance $\Delta \widehat{a}_i$, which in 21 cases is too small for causing a significant deviation in any coupling node. In the test series \texttt{attack\_3}, an attack is detected in all 100 time steps. In these 79 respectively 100 time steps, the attack identification method is applied. For both experiments, Table~\ref{tab: identification results} lists how often the actual, unknown attack set $\text{supp}(\widehat{a})$ or a superset is correctly identified, and how often the sufficient condition of Theorem~\ref{thm: superset identification} resp.\ Theorem~\ref{thm: correct identification} is satisfied. The results of \texttt{attack\_1} are shown in tables (a) and (b), those of \texttt{attack\_3} in tables (c) and (d). The left tables refer to identifying a superset of $\text{supp}(\widehat{a})$ as in Theorem~\ref{thm: superset identification}, the right tables to identifying the attack set exactly as in Theorem~\ref{thm: correct identification}.
\begin{table}[h]
	\caption{Fourfold tables showing the results of experiments \texttt{attack\_1} (tables (a) and (b)) and \texttt{attack\_3} ((c) and (d)) with one respectively three random attackers per time step. 
	}
	\label{tab: identification results}
	\small
	\begin{subtable}[t]{0.4\columnwidth}
		\caption{Superset identification according to Theorem~\ref{thm: superset identification}}
		\begin{tabular}{c|a|c}
			\scriptsize{\texttt{attack\_1}}& Ident. & $\overline{\text{Ident.}}$ \\ 
			\hline \rowcolor{matplotlibGray!60}\rule[.25ex]{0pt}{2.5ex}
			Cond. & 94.94\% & 0.00\% \\ 
			\hline \rule[.25ex]{0pt}{2.5ex}
			$\overline{\text{Cond.}}$ & 5.06\% & 0.00\% \\ 
			\hline \rule[.25ex]{0pt}{2.5ex}
			& 100\% &
		\end{tabular} 
		\vskip1ex
	\end{subtable} 
	\hfil
	\begin{subtable}[t]{0.4\columnwidth}
		\caption{Exact identification \mbox{according} to Theorem~\ref{thm: correct identification}}
		\begin{tabular}{c|a|c}
			\scriptsize{\texttt{attack\_1}}		& Ident. & $\overline{\text{Ident.}}$ \\ 
			\hline \rowcolor{matplotlibBlue!60}\rule[.25ex]{0pt}{2.5ex}		
			Cond. & 93.67\% & 0.00\% \\ 
			\hline \rule[.25ex]{0pt}{2.5ex}
			$\overline{\text{Cond.}}$ & 6.33\% & 0.00\% \\ 
			\hline \rule[.25ex]{0pt}{2.5ex}
			& 100\% &\\
		\end{tabular} 
		\vskip1ex
	\end{subtable}
	\vskip.4cm
	\begin{subtable}[t]{0.4\columnwidth}
		\caption{Superset identification according to Theorem~\ref{thm: superset identification}}
		\begin{tabular}{c|a|c}
			\scriptsize{\texttt{attack\_3}}		& Ident. & $\overline{\text{Ident.}}$ \\ 
			\hline \rowcolor{matplotlibGray!60}\rule[.25ex]{0pt}{2.5ex}
			Cond. & 40.00\% & 0.00\% \\ 
			\hline \rule[.25ex]{0pt}{2.5ex}
			$\overline{\text{Cond.}}$ & 59.00\% & 1.00\% \\ 
			\hline \rule[.25ex]{0pt}{2.5ex}
			& 99.00\% &	
		\end{tabular} 
	\end{subtable}
	\hfil
	\begin{subtable}[t]{0.4\columnwidth}	
		\caption{Exact identification \mbox{according} to Theorem~\ref{thm: correct identification}}
		\begin{tabular}{c|a|c}
			\scriptsize{\texttt{attack\_3}}		& Ident. & $\overline{\text{Ident.}}$ \\ 
			\hline \rowcolor{matplotlibBlue!60}\rule[.25ex]{0pt}{2.5ex}
			Cond. & 31.00\% & 0.00\% \\ 
			\hline \rule[.25ex]{0pt}{2.5ex}
			$\overline{\text{Cond.}}$ & 51.00\% & 18.00\%\\ 
			\hline \rule[.25ex]{0pt}{2.5ex}
			& 82.00\% &	
		\end{tabular} 
	\end{subtable}
	\vskip.4cm
	\caption*{\fbox{\parbox{0.95\columnwidth}{{\small Cond. = Sufficient condition satisfied, $\overline{\text{Cond.}}$ = not satisfied \\ Ident. = (Superset of) $\text{supp}(\widehat{a})$ identified, $\overline{\text{Ident.}}$ = not identified}}}}
\end{table}

Considering the experiments \texttt{attack\_1}, the green highlighted column of Table~\ref{tab: identification results}\,(a) reveals that at each time the identification method is applied, it correctly identifies a superset of the unknown attack set. In 94.94\% of the cases, this is guaranteed since the sufficient condition of Theorem~\ref{thm: superset identification} is satisfied and implies the correct identification of a superset. In 5.06\%, however, the condition is not fulfilled but still some superset is computed. This is possible because the theorem only states a sufficient but not necessary condition. 
Since $\widetilde{\delta} < \delta$ with $\delta, \widetilde{\delta}$ denoting the parameters occurring in Theorems~\ref{thm: superset identification} and~\ref{thm: correct identification}, the sufficient condition of Theorem~\ref{thm: correct identification} is harder to fulfill than the one of Theorem~\ref{thm: superset identification}. This is reflected in Table~\ref{tab: identification results}\,(b), showing that the sufficient condition of Theorem~\ref{thm: correct identification} is satisfied in 93.67\%, in contrast to 94.94\% in Table~\ref{tab: identification results}\,(a).
The exact identification is successful at all times, although in 6.33\% this is not guaranteed by Theorem~\ref{thm: correct identification}.

\addtolength{\textheight}{-.2cm}   

The sufficient conditions in both theorems become harder to satisfy the larger $\|\Delta \widehat{a}\|_1 + M\|\Delta \widehat{z}\|_1$ gets, where $\Delta \widehat{a}$, $\Delta \widehat{z}$ denote the \mbox{occurring} attack and the caused coupling deviations, and $M$ is the maximum degree in the subsystem network. Since in the test series \texttt{attack\_3} three inputs per time step are randomly disturbed in contrast to only one in \texttt{attack\_1}, the resulting values $\|\Delta \widehat{a}\|_1$, $\|\Delta \widehat{z}\|_1$ are expected to be larger. 
This becomes evident in the comparison of Tables~\ref{tab: identification results}\,(a) with (c), and (b) with (d), respectively. The sufficient condition of Theorem~\ref{thm: superset identification} (highlighted in gray) as well as Theorem~\ref{thm: correct identification} (blue) is fulfilled in significantly fewer cases. 
In more than 98\% resp.\ 73\% of all cases with unfulfilled sufficient condition, however, a superset resp.\ the attack set $\text{supp}(\widehat{a})$ itself are still correctly identified, such that total scores of 99\% for superset identification and 82\% for exact identification are reached.
Attacking three out of 18 inputs (corresponding to the size of the reduced sensitivity matrices $\widetilde{S}_I^a$), means compromising more than~15\% of the system simultaneously and thus requires attackers with very powerful resources. In this context, the achieved success rates should be regarded as very high.

Setting up the relaxed identification problem~\eqref{opt: identification problem relaxed} requires the parameter $\varepsilon$, which depends on the unknown attack~$\Delta \widehat{a}$, such that computing a solution $\Delta a^\ast_{\eqref{opt: identification problem relaxed}}$ to identify the attack set $\text{supp}(\widehat{a})$ exactly is a rather theoretical consideration or requires a good estimate of $\varepsilon$. For the actual application as attack identification method, solving the identification problem~\eqref{opt: identification problem} is more suitable and guaranteed to find a superset $\text{supp}(a^\ast_{\eqref{opt: identification problem}}) \supseteq \text{supp}(\widehat{a})$ under the condition of Theorem~\ref{thm: superset identification}. The set $\text{supp}(a^\ast_{\eqref{opt: identification problem}}) \setminus \text{supp}(\widehat{a})$, containing the wrongly identified inputs, on average contains 0.56 indices in the test series \texttt{attack\_1} and 0.9 in \texttt{attack\_3}. In a more realistic scenario, we find it valid to assume that the attack set $\text{supp}(\widehat{a})$ remains constant for some time and the attack set must not be identified within only one sampling time. On the contrary, it seems very promising that already within one time step a superset containing all attacked inputs is identified with very high success rate. Even if one or two benign inputs are contained, one can use the findings over several time steps to draw a sophisticated conclusion about the actual attack set.

\section{CONCLUSION}
We considered a hierarchical method for attack identification in distributed nonlinear control systems from preliminary work and carried out a detailed analysis in terms of both theoretical guarantees and numerical results. The method is based on the exchange of locally evaluated sensitivity information and solves a sparse signal recovery problem at each time step. It allows to identify arbitrary attacks on the system's inputs, without requiring global model knowledge nor assuming any attack patterns to be known. We derived sufficient conditions depending on the strength of the attack and properties of the system's dynamics, under which the method is guaranteed to identify all attacked inputs. Numerical experiments for the identification of faulty buses in the IEEE~30 bus power system revealed that not only the sufficient conditions are largely met, but also the success rates of correct identification are very high, although a very demanding attack scenario was considered with randomly generated attacks changing at each time step.

\bibliographystyle{../Template/IEEEtran}
\bibliography{../references_cdc2020}

\end{document}